%% file: main.tex
\documentclass{article}

\usepackage{spconf}
\usepackage[T1]{fontenc}
\usepackage{amsmath,amssymb,mathtools,amsthm}
\usepackage{bm}
\usepackage{booktabs}
\usepackage{cite}
\usepackage{graphicx}
\usepackage{microtype}
\usepackage{tikz}
\usetikzlibrary{arrows.meta,positioning,shapes.geometric}
\usepackage[hidelinks]{hyperref}

\newcommand{\CC}{\mathbb C}
\newcommand{\RR}{\mathbb R}
\newcommand{\HH}{\mathsf H}
\newcommand{\TT}{\mathsf T}

\newcommand{\CN}{\mathcal{CN}}
\newcommand{\Tr}{\operatorname{Tr}}
\newcommand{\diag}{\operatorname{diag}}
\newcommand{\col}{\operatorname{col}}
\newcommand{\rank}{\operatorname{rank}}

\newcommand{\calI}{\mathcal I}

\newtheorem{proposition}{Proposition}

\newtheorem{remark}{Remark}

\title{Communication-Centric Symbol-Level Precoding for\\
Fronthaul-Limited Networked ISAC}

\name{\shortstack{Shu Cai$^{\star}$, Amin Li$^{\star}$,
Yiming Zhang$^{\star}$, Jun Zhang$^{\star}$, Qi Zhang$^{\star}$, and Ya-Feng Liu$^{\dagger}$}
}
\address{ $^{\star}$College of Telecommunications and Information Engineering, \\ Nanjing University of Posts and Telecommunications, Nanjing, China\\ 
$^{\dagger}$School of Mathematical Sciences, Beijing University of Posts and Telecommunications, Beijing, China}

\begin{document}

\ninept
\maketitle
\begin{abstract}
This paper studies communication-centric symbol-level precoding (SLP)
for fronthaul-limited networked integrated sensing and communication
(ISAC). SLP is attractive for exploiting
constructive interference, but its data-dependent nature makes fronthaul-aware waveform
design difficult because the transmit covariance used in conventional
linear-precoding compression models is unavailable before the symbol
block is formed. To address this issue, we develop a compression-aware
SLP framework that relates the fronthaul requirement to the realized
block energy. We further formulate a sensing-oriented joint design. To
solve the resulted high-dimensional and nonconvex problem, we
propose a tractable lower-bound method with a performance bound.
Simulations show that the proposed SLP
scheme improves sensing performance over linear-beamforming baselines
under fronthaul constraints.
\end{abstract}


\input{sections/00_introduction}
\input{sections/01_system_model}
\input{sections/02_problem_formulation}
\input{sections/03_problem_reformulation}
\input{sections/03_slbm_cvx}
\input{sections/03_sdr_upper_bound}
\input{sections/06_simulations}

\bibliographystyle{IEEEbib}
\bibliography{references}

\end{document}

%% file: sections/00_introduction.tex
\section{Introduction}
\label{sec:introduction}

Integrated sensing and communication (ISAC) is a key physical-layer
direction for future wireless networks, enabling shared radio
infrastructure to serve users and sense the environment, while
networked ISAC coordinates distributed transmitters through a central
processor (CP) for spatial diversity and coherent sensing gains
\cite{liu2020joint,liu2022integrated,zhang2025networked}.

Early ISAC waveform and beamforming studies characterized the
sensing--communication tradeoff under co-located or coordinated
architectures, often assuming non-capacity-limited transport
\cite{liu2017robust,liu2020joint,liu2022integrated,xu2023coordinated,
li2024networked,zhao2024joint,liu2022cramer}. For practical
networked ISAC, the CP sends downlink waveforms to transmitters and
receives compressed sensing observations, so fronthaul compression
affects both the radiated waveform and sensing noise. Compression-aware
transmission has been studied in cloud radio access networks (C-RANs)
and compressing-relay networks
\cite{peng2015fronthaul,zhang2021fronthaul,liu2021uplinkdownlink,fan2025qos},
and the works \cite{zhang2025networked,zhang2025optimal} introduced capacity-limited fronthaul into networked ISAC by jointly optimizing compression and block-level precoding (BLP).

Symbol-level precoding (SLP) adapts waveforms to instantaneous
data symbols and steers interference into
constructive-interference (CI) regions
\cite{masouros2015constructive,liTutorialInterference2020a}. Recent ISAC
studies show that SLP can improve joint sensing-communication performance by
reusing communication symbols as waveform degrees of freedom
\cite{chen2024symbol,ShuDistributedADMMBased2026}. Motivated by these
advantages, we apply SLP to fronthaul-limited networked ISAC.


This extension faces two main challenges. First, the data-dependent SLP waveform and the downlink compression-noise variance are circularly coupled: SLP design requires the compression-noise variance, whereas this variance depends on the realized SLP block energy. Second, jointly optimizing the SLP waveforms and fronthaul compression variables leads to a high-dimensional nonconvex problem spanning all antennas and symbol intervals.

The main contributions of this paper are threefold. First, we derive a distribution-free upper bound on the SLP fronthaul rate using Gaussian maximum entropy~\cite{cover2006elements} and instantiate it with the realized block-average waveform energy. The bound requires no prescribed transmit covariance and reduces to the conventional BLP rate formula in~\cite{zhang2025networked} for temporally white Gaussian data. 
Second, we formulate a compression-aware sensing signal-to-interference-plus-noise ratio (SINR) maximization problem and recast it as an equivalent reduced second-order-cone (SOC)-constrained form. Third, we develop an SLP successive lower-bound maximization (SLP-SLBM) algorithm that yields a feasible lower-bound design and an SOC-strengthened SLP semidefinite relaxation (SLP-SDR) that provides a global upper bound, allowing the solution quality to be assessed through their gap. Simulations show that SLP-SLBM performance remains close to this upper bound and outperforms the BLP baselines.


%% file: sections/01_system_model.tex
\section{System and Signal Model}
\label{sec:system}

We consider the communication-centric networked ISAC architecture in
Fig.~\ref{fig:system}, following \cite{zhang2025networked}. A CP coordinates \(L\) distributed ISAC transmitters (TXs),
one sensing receiver (RX), and \(K\) single-antenna communication users.
The CP generates a block of
communication-centric SLP waveforms, conveys the waveforms to the TXs
over downlink fronthaul links, and jointly exploits the resulting
transmissions for user service and target illumination. The sensing RX
compresses its target echo and returns it to the CP over an uplink
fronthaul link for coherent detection.

\begin{figure}[t]
  \centering
  \includegraphics[width=0.83\columnwidth]{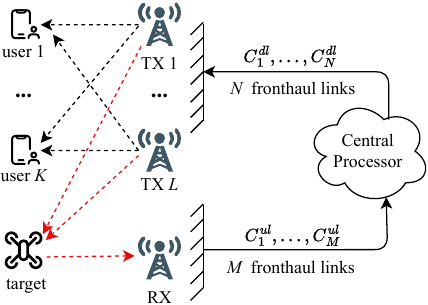}
  \caption{The fronthaul-limited networked ISAC.}
  \label{fig:system}
\end{figure}
\subsection{Communication Signal Model}
Each TX has \(N_t\) antennas, so the total number of transmit antennas
is \(N=LN_t\). For a block of \(T\) known \(\mathcal{M}\)-ary phase-shift keying
(PSK) symbol vectors, the CP designs the
symbol-level transmit vector
\(\bm{x}(t)\in\CC^N\) and collects
\(\bm X\triangleq[\bm x(1),\ldots,\bm x(T)]\in\CC^{N\times T}\) \cite{liTutorialInterference2020a}. The CP compresses each waveform sample before sending it to the corresponding TX antenna. We model this
downlink fronthaul operation by the Gaussian test channel \cite{zhang2025networked,zhang2025optimal}
\begin{equation}
  \bar{\bm{x}}(t)=\bm{x}(t)+\bm e^{\rm dl}(t),\qquad
  \bm e^{\rm dl}(t)\sim\CN(\bm0,\bm Q_{\rm dl}),
  \label{eq:dl-test}
\end{equation}
where \(\bm Q_{\rm dl}=\diag(q_1^{\rm dl},\ldots,q_N^{\rm dl})\).
The covariance of the actually radiated signal is therefore
\begin{align}
  \bm R
    = \bm X\bm X^\HH/T+\bm Q_{\rm dl}.
  \label{eq:R}
\end{align}
Let \(\bm h_k\in\CC^N\) be the channel from all TX antennas to user
\(k\) and \(v_k(t)\sim\CN(0,\sigma_v^2)\) the receiver noise. The received sample
is
\begin{equation}
  y_k(t)
  =\bm h_k^\HH\bm x(t)
   +\bm h_k^\HH\bm e^{\rm dl}(t)+v_k(t).
  \label{eq:user-rx}
\end{equation}
The CP controls the noiseless term \(\bm h_k^\HH\bm x(t)\)  and the random downlink compression term contributes the variance \(\bm h_k^\HH\bm Q_{\rm dl}\bm h_k\). 
\subsection{Radar Signal Model}
For a prescribed angle--range cell,
we adopt the rank-one line-of-sight target channel \cite{zhang2025networked}
\begin{equation}
  \bm G
  =\bm a_r\bm a_t^\HH\bm\Sigma_g,\quad
  \bm\Sigma_g
  =\diag(g_1,\ldots,g_L)\otimes\bm I_{N_t},
  \label{eq:G}
\end{equation}
where \(\bm a_t\) and \(\bm a_r\) are the transmit and receive steering
vectors. The sensing RX observes
\begin{equation}
  \bm y^{\rm s}(t)
  =\bm G\bar{\bm x}(t)+\bm z(t),\quad
  \bm z(t)\sim\CN(\bm0,\sigma_z^2\bm I_M),
  \label{eq:sensing-rx}
\end{equation}
where $\bm z(t)$ is the receiver noise.
Before forwarding the echo to the CP, the RX applies an uplink Gaussian
test channel,
\begin{align}
  \bar{\bm y}^{\rm s}(t)
    =\bm y^{\rm s}(t)+\bm e^{\rm ul}(t),\quad \bm e^{\rm ul}(t)
    \sim\CN(\bm0,\bm Q_{\rm ul}),
  \label{eq:Qul}
\end{align}
where \(\bm Q_{\rm ul}=\diag(q_1^{\rm ul},\ldots,q_M^{\rm ul})\).
Thus, the effective sensing-noise covariance is
\(\bm R_{\rm eff}=\bm Q_{\rm ul}+\sigma_z^2\bm I_M\).



%% file: sections/02_problem_formulation.tex
\section{Problem Formulation}
\label{sec:problem}

This section first specifies the communication and radar performance
metrics and then states the joint design problem.

\subsection{Compression-Aware Communication SLP Metric}

Rotating the noiseless received sample by the
desired PSK symbol \(s_k(t)\) gives
  $z_k(t)\triangleq \bm h_k^\HH\bm x(t)s_k^*(t).$
Following the standard constructive-interference geometry of SLP, the compression-aware decision-region constraint is \cite{liTutorialInterference2020a,ShuDistributedADMMBased2026}
\begin{equation}
  |\Im\{z_k(t)\}|
  \le
  \left(\Re\{z_k(t)\}
  -\sqrt{\Gamma_k\!\big(
    \bm h_k^\HH\bm Q_{\rm dl}\bm h_k+\sigma_v^2
  \big)}\right)\tan\phi,
  \label{eq:ci-region}
\end{equation}
where \(\Gamma_k\) is the communication SINR target and
\(\phi=\pi/\mathcal{M}\).
The square-root term accounts jointly for downlink compression noise and
receiver noise.

\subsection{Radar Detection Metric}

For a spatial combiner \(\bm w\), define the
detection SINR
\begin{equation}
  \gamma_{\rm s}(\bm w)
  = 
  (\bm w^\HH\bm G\bm R\bm G^\HH\bm w)
       (\bm w^\HH\bm R_{\rm eff}\bm w)^{-1}.
  \label{eq:gamma-w}
\end{equation}
By the rank-one target channel in \eqref{eq:G},
  $\bm G\bm R\bm G^\HH
  =\calI(\bm R)\bm a_r\bm a_r^\HH$,
where \(\calI(\bm R)\) denotes the target illumination, i.e.,
\begin{equation}
  \calI(\bm R)
  =
  \bm a_t^\HH\bm\Sigma_g\bm R\bm\Sigma_g^\HH\bm a_t.
  \label{eq:illumination-R}
\end{equation}
Thus, \eqref{eq:gamma-w} is maximized by 
\(\bm w^\star\propto\bm R_{\rm eff}^{-1}\bm a_r\), yielding
 the optimized sensing SINR
\begin{equation}
  \operatorname{SINR}^{\rm s}
  =\calI(\bm R)\,
   \bm a_r^\HH
   (\bm Q_{\rm ul}+\sigma_z^2\bm I_M)^{-1}\bm a_r.
  \label{eq:sinr-opt}
\end{equation}
The conventional point-target detection model below motivates
\eqref{eq:sinr-opt} as the radar objective.

\begin{remark}[Conventional radar interpretation]
\label{rem:radar-meaning}
For a fixed false-alarm probability \(P_{\rm FA}\), the
detection probability is
  $ P_{\rm D}
  =1-F_{\chi_2^2(\rho_T)} (
    F_{\chi_2^2}^{-1}(1-P_{\rm FA}))$,
where \(F_{\chi_2^2}\) and \(F_{\chi_2^2(\rho_T)}\) are the cumulative distribution functions (CDFs) of
central and noncentral chi-squared variables with two degrees of freedom  \cite{liu2017robust}.
 Since
\(P_{\rm D}\) increases monotonically with \(\rho_T\) and \(\rho_T\propto T\gamma_{\rm s}(\bm w)\)
\cite{bekkerman2006target}, increasing
\eqref{eq:gamma-w}, and its optimized value \eqref{eq:sinr-opt},
improves the point-target detection probability.
\end{remark}

\input{sections/02a_fronthaul_rate_model}

\subsection{Joint Design Problem}

For common per-stream capacities \(C^{\rm dl}\) and \(C^{\rm ul}\), the considered
communication-centric ISAC joint design problem is
\begingroup
\setlength{\jot}{1pt}
\begin{subequations}
\label{prob:original}
\begin{align}
  \underset{\substack{\bm X,\{q_n^{\rm dl}\ge0\},\\
                       \{q_m^{\rm ul}\ge0\}}}
  {\operatorname{max}}~&
  \operatorname{SINR}^{\rm s}
  \label{prob:original-obj}\\
  \operatorname{s.t.}~\quad&
  \eqref{eq:ci-region},\notag \\
  &\log_2\!\big(
    1+\|\bm x_n\|_2^2/(Tq_n^{\rm dl})
  \big)
  \le C^{\rm dl},\quad \forall n,
  \label{eq:dl-fh}\\
  &\log_2\!\big(
    1+{(\calI(\bm R)/M+\sigma_z^2)}/{q_m^{\rm ul}}
  \big)
  \le C^{\rm ul},\quad \forall m,
  \label{eq:ul-fh}\\
  &\|\bm X\|_F^2/T
   +\sum_{n=1}^{N}q_n^{\rm dl}
  \le P_{\max}.
  \label{eq:power-constraint}
\end{align}
\end{subequations}
\endgroup
The constraints enforce, respectively, the CI regions for SLP, the per-stream
downlink and uplink fronthaul capacities \cite{liu2021uplinkdownlink,zhang2025networked}, and
the total transmit-power budget. The uplink rate uses the fact
\(a_{r,m}=\tfrac{1}{\sqrt{M}} e^{-\jmath\pi\cos\theta_r}\) with \(\theta_r\) being the direction of arrival.

%% file: sections/02a_fronthaul_rate_model.tex
\subsection{Downlink Fronthaul Rate Model}
\label{subsec:time-average}

The downlink test channel in \eqref{eq:dl-test} requires a fronthaul
rate no smaller than \(T^{-1}I(\bm x_n;\bar{\bm x}_n)\) for antenna
\(n\), where \(\bm x_n=[x_n(1),\ldots,x_n(T)]^\TT\) is the source
block, \(\bm e_n\sim\CN(\bm0,q_n^{\rm dl}\bm I_T)\) is independent
of \(\bm x_n\), and \(\bar{\bm x}_n=\bm x_n+\bm e_n\). Unlike
linear precoding, the SLP block is data-dependent and has no prescribed
distribution. We therefore use the following distribution-free upper
bound.

\begin{proposition}[Distribution-free block-rate bound]
\label{prop:block-rate}
For every distribution of \(\bm x_n\) with finite second moment, we have
\begin{align}
  \frac{1}{T}I(\bm x_n;\bar{\bm x}_n)
  &\le \frac{1}{T}\log_2\det\!\Big(
    \bm I_T+\frac{\bm K_n}{q_n^{\rm dl}}\Big) \le\! \log_2\!\Big(
    1+\frac{\mathbb E\|\bm x_n\|_2^2}
             {Tq_n^{\rm dl}}\Big),
  \label{eq:rate-energy-bound}
\end{align}
where \(\bm K_n=\operatorname{Cov}(\bm x_n)\). Both inequalities are
equalities when \(\bm x_n\) is zero-mean proper Gaussian with covariance
\(P_n\bm I_T\).
\end{proposition}

\begin{proof}
The first inequality follows from the Gaussian maximum-entropy bound
\cite{cover2006elements}; the second from the
arithmetic-geometric mean inequality and
\(\Tr(\bm K_n)\le\mathbb E\|\bm x_n\|_2^2\).
\end{proof}


For a realized SLP block, we use
  $\hat P_{n,T}\triangleq \|\bm x_n\|_2^2/T$
to approximate \(\mathbb E\|\bm x_n\|_2^2/T\). This sample average
converges to the ensemble power for a stationary ergodic source as
\(T\) grows.

%% file: sections/03_problem_reformulation.tex
\section{Problem Reformulation and Simplification}
\label{sec:reformulation}

Problem \eqref{prob:original} is nonconvex due to the quadratic term in the objective. We first reformulate it into a more tractable form.

\subsection{Elimination of Uplink Compression}

Define $\beta = (2^{C^{\rm ul}}-1)^{-1}$.
For fixed \(\bm R\), the objective \eqref{eq:sinr-opt} decreases in every
\(q_m^{\rm ul}\). Hence all uplink constraints in \eqref{eq:ul-fh} are
active at an optimum, and
\begin{equation}
  q_m^{{\rm ul},\star}
  =\beta\bigl(\calI(\bm R)/M+\sigma_z^2\bigr),
  \quad m=1,\ldots,M.
  \label{eq:qul-star}
\end{equation}
Thus, 
\(\bm Q_{\rm ul}^{\star}=q^{{\rm ul},\star}\bm I_M\), and the 
detection SINR becomes
\begin{align}
  \operatorname{SINR}^{\rm s}
    &= \frac{\calI(\bm R)}
      {(\beta/M)\calI(\bm R)+(1+\beta)\sigma_z^2},
      \quad \calI(\bm R)\ge0.
  \label{eq:fI}
\end{align}
This expression is strictly increasing in \(\calI(\bm R)\), so the
design exactly reduces to illumination maximization. The uplink fronthaul imposes the asymptotic ceiling
\(\lim_{\calI\rightarrow\infty}\operatorname{SINR}^{\rm s}=M/\beta\).

\subsection{Equivalent Reduced SOC Formulation}

Set \(\alpha_{\rm dl}=2^{C^{\rm dl}}-1\),
\(p_n=\sqrt{q_n^{\rm dl}}\), and let 
\(\bm p=[p_1,\ldots,p_N]^\TT\).
The downlink fronthaul constraint \eqref{eq:dl-fh} is equivalent to
\begin{equation}
  \|\bm x_n\|_2
  \le\sqrt{T\alpha_{\rm dl}}\,p_n,
  \qquad n=1,\ldots,N.
  \label{eq:fh-soc}
\end{equation}

For user \(k\), defining $g_{k,t}^{\pm}(\bm x(t))
  =
  \Re\{z_k(t)\}\pm\Im\{z_k(t)\}\cot\phi$ and 
\(\bm D_k=\sqrt{\Gamma_k}\diag(|\bm h_k|)\), then $\Gamma_k(\bm h_k^\HH\bm Q_{\rm dl}\bm h_k+\sigma_v^2) = \|\col(\bm D_k\bm p,\\\sqrt{\Gamma_k}\sigma_v)\|_2^2$. 
The constraint \eqref{eq:ci-region} is exactly
\begin{align}\label{eq:ci-soc}
  \|\col(\bm D_k\bm p,\sqrt{\Gamma_k}\sigma_v)\|_2
    \le g_{k,t}^{\pm}(\bm x(t)),
  \quad \forall k,t.
\end{align}
Let \(\bm b=\bm\Sigma_g^\HH\bm a_t\) and
\(\bm D_b=\diag(|b_1|^2,\ldots,|b_N|^2)\). 
Substitution of \eqref{eq:R} into \eqref{eq:illumination-R} gives the
convex quadratic illumination
\begin{equation}
  \calI(\bm X,\bm p)
  =\|\bm X^\HH\bm b\|_2^2/T
   +\bm p^\TT\bm D_b\bm p.
  \label{eq:I-Xp}
\end{equation}
For compactness, write \(\bm v(\bm X,\bm p)=
\col(\operatorname{vec}(\bm X)/\sqrt T,\bm p)\). Problem \eqref{prob:original} is equivalent to
\begin{subequations}
  \label{prob:pred}
\begin{align}
  \max_{\bm X,\bm p,\bm u}\quad&\calI(\bm X,\bm p)\\
  \operatorname{s.t.}\quad
  &\|\bm x_n\|_2\le\sqrt{T\alpha_{\rm dl}}\,p_n,
    \quad \forall n,\label{prob:pred-fh}\\
  &\left\|\col(\bm D_k\bm p,\sqrt{\Gamma_k}\sigma_v)\right\|_2
    \le u_k,\quad \forall k,\label{prob:pred-cm}\\
  &u_k\le g_{k,t}^{\pm}(\bm x(t)),
    \quad \forall k,t, \label{prob:pred-ci}\\
  &\|\bm v(\bm X,\bm p)\|_2
    \le\sqrt{P_{\max}}.\label{prob:pred-power}
    \end{align}
\end{subequations}
Note that the original CI constraints \eqref{eq:ci-soc} imply \eqref{prob:pred-cm} --
\eqref{prob:pred-ci} by choosing \(u_k=\|\col(\bm D_k\bm p,\sqrt{\Gamma_k}\sigma_v)\|_2\), and
the converse is immediate. Thus, \eqref{prob:pred} is equivalent to
\eqref{prob:original}, but replaces \(2KT\) SOCs by \(K\) SOCs and
\(2KT\) halfspaces.

%% file: sections/03_slbm_cvx.tex
\section{Proposed Solution Approach}
\label{sec:solution}
This section develops two complementary solution methods for
\eqref {prob:pred}: an SLBM-based feasible lower-bound design and an
SDR-based upper bound.
\subsection{SLBM-Based Lower-Bound Design}
\label{sec:slbm}
We first seek a feasible solution to the reduced problem
\eqref{prob:pred}. Its constraints are convex, but the objective
\eqref{eq:I-Xp} is a convex quadratic function whose maximization is
nonconvex. We therefore employ SLBM to construct a sequence of feasible
designs. At a feasible point \((\bm X^{(i)},\bm p^{(i)})\),
first-order convexity gives the affine global lower bound
\begin{align}
  \widehat{\calI}^{(i)}(\bm X,\bm p)
  \triangleq{}&\frac{2}{T}\Re\!\left\{
  \Tr[(\bm X^{(i)})^\HH\bm b\bm b^\HH\bm X]\right\}
  -\frac{1}{T}\|(\bm X^{(i)})^\HH\bm b\|_2^2\notag\\
  &+2(\bm p^{(i)})^\TT\bm D_b\bm p
  -(\bm p^{(i)})^\TT\bm D_b\bm p^{(i)}
  \le\calI(\bm X,\bm p),
\end{align}
which is tight and gradient consistent at the current point. With
\(\mathcal F_{\rm red}\) denoting the feasible set of \eqref{prob:pred},
the next iterate is obtained from
\begin{equation}
  (\bm X^{(i+1)},\bm p^{(i+1)},\bm u^{(i+1)})
  \in\arg\max_{(\bm X,\bm p,\bm u)\in\mathcal F_{\rm red}}
  \widehat{\calI}^{(i)}(\bm X,\bm p).
\end{equation}
This subproblem is a second-order cone program (SOCP) and is solved by CVX \cite{cvx}. Since the
surrogate is a global lower bound and is tight at the current point,
each update preserves feasibility and monotonically improves the
original objective.

The required feasible starting point is obtained by removing the power
constraint and solving the minimum-power SOCP
\begin{equation}
  P_0\triangleq
  \min_{\bm X,\bm p,\bm u}
  \left\{T^{-1}\|\bm X\|_F^2+\|\bm p\|_2^2:
  \eqref{prob:pred-fh}\text{--}\eqref{prob:pred-ci}\right\}.
\end{equation}
If this problem is infeasible or \(P_0>P_{\max}\), then
\eqref{prob:pred} is infeasible; otherwise, its optimizer initializes
\((\bm X^{(0)},\bm p^{(0)},\bm u^{(0)})\).

SLBM thus returns a feasible value satisfying
\(\calI_{\rm SLBM}\le\calI^\star\). However, because the original
maximization is nonconvex, the attained lower bound may depend on the
initialization and, by itself, does not reveal its distance from the
global optimum. This motivates the SDR upper bound developed next, whose gap from
\(\calI_{\rm SLBM}\) quantifies the quality of the feasible design.

%% file: sections/03_sdr_upper_bound.tex
\subsection{SDR-Based Upper Bound}
\label{sec:sdr}
To assess the quality of the SLBM design, we recast \eqref{prob:pred}
as an exact homogeneous quadratically constrained quadratic program
(QCQP), derive its rank-one lift, and relax the rank constraint to
obtain a computable global upper bound \cite{liu2024survey}.

\emph{1) Exact homogeneous QCQP and rank-one lift:}

Define \(\widetilde{\bm\xi}_t=T^{-1/2}
\col(\Re\{\bm x(t)\},\Im\{\bm x(t)\})\), 
\(\bm y=\col(\widetilde{\bm\xi}_1,\ldots,
\\\widetilde{\bm\xi}_T,\bm p)\in\RR^d\), where \(d=2NT+N\).
Let \(\bm A_n\) select antenna \(n\)'s \(2T\) real samples,
\(\bm e_n^p\) select \(p_n\), and \(\bm a_{k,t}^{\nu}\) satisfy
\((\bm a_{k,t}^{\nu})^\TT\bm y=g_{k,t}^{\nu}(\bm x(t))\) with $\nu=\pm$.
Set \(\bm F_n=\bm A_n^\TT\bm A_n-\alpha_{\rm dl}
\bm e_n^p(\bm e_n^p)^\TT\),
\(\bm H_k=\operatorname{blkdiag}
(\bm0_{2NT},\bm D_k^\TT\bm D_k)\), and
\(\bm C_{k,t}^{\nu}=\bm H_k-\bm a_{k,t}^{\nu}
(\bm a_{k,t}^{\nu})^\TT\). The physical constraints become
\begin{subequations}
\label{eq:sdr-qc}
\begin{align}
  \bm y^\TT\bm F_n\bm y&\le0,\quad
  (\bm e_n^p)^\TT\bm y\ge0, \quad \forall n,\\
  \bm y^\TT\bm C_{k,t}^{\nu}\bm y
    +\Gamma_k\sigma_v^2&\le0,\quad (\bm a_{k,t}^{\nu})^\TT\bm y \ge0, \quad \forall k,t,\nu,\\
  \bm y^\TT\bm y&\le P_{\max}.
\end{align}
\end{subequations}
The sign constraints prevent squaring from admitting the wrong CI
halfspace. With the complex-to-real conversion operator
\(\mathcal T(\bm A)=
\left[\begin{smallmatrix}\Re\bm A&-\Im\bm A\\
\Im\bm A&\Re\bm A\end{smallmatrix}\right]\), define
\(\bm Q_y=\operatorname{blkdiag}
(\mathcal T(\bm b\bm b^\HH),\ldots,
\mathcal T(\bm b\bm b^\HH),\bm D_b)\succeq\bm0\). Then
\(\calI(\bm X,\bm p)=\bm y^\TT\bm Q_y\bm y\), so
\eqref{eq:sdr-qc} is an exact homogeneous QCQP.

Homogenize with \(\bar{\bm y}=\col(\bm y,1)\) and
\(\bm Y=\bar{\bm y}\bar{\bm y}^\TT\). Define
\(\widehat{\bm Q}=\operatorname{blkdiag}(\bm Q_y,0)\),
\(\widehat{\bm F}_n=\operatorname{blkdiag}(\bm F_n,0)\),
\(\widehat{\bm C}_{k,t}^{\nu}=\operatorname{blkdiag}
(\bm C_{k,t}^{\nu},\Gamma_k\sigma_v^2)\), and
\(\widehat{\bm P}=\operatorname{blkdiag}(\bm I_d,-P_{\max})\).
Writing \(\bm y_Y=\bm Y_{1:d,d+1}\) and letting
\(\bm p_Y\) be its last \(N\) entries, the equivalent rank-one
lift of \eqref{prob:pred} is
\begin{subequations}
\label{eq:rank-lift}
\begin{align}
  \max_{\bm Y}\quad&\langle\widehat{\bm Q},\bm Y\rangle\\
  \operatorname{s.t.}\quad&
  \langle\widehat{\bm F}_n,\bm Y\rangle \le0, \quad (\bm e_n^p)^\TT\bm y_Y \ge0, \quad \forall n,\label{eq:rank-lift1}\\
  &\langle\widehat{\bm C}_{k,t}^{\nu},\bm Y\rangle \le0, ~ 
    (\bm a_{k,t}^{\nu})^\TT\bm y_Y \ge0,
    ~~  \forall k,t,\nu,\label{eq:rank-lift2}\\
  &\langle\widehat{\bm P},\bm Y\rangle \le0,\label{eq:rank-lift3}\\
  & Y_{d+1,d+1} =1,\quad \bm Y\succeq\bm0,\quad
  \rank(\bm Y)=1,\label{eq:rank-lift4}
\end{align}
\end{subequations}
where \(\langle\bm A,\bm B\rangle\) equals trace of \(\bm A^\TT\bm B\).

\emph{2) SOC-strengthened SDR:}

Simply dropping \(\rank(\bm Y)=1\) gives the standard SDR. For a
high-rank \(\bm Y\), however, its lifted quadratic
constraints restrict second-order moments but need not enforce the
original SOC geometry on the first-order vector \(\bm y_Y\).
We therefore strengthen the relaxation by directly retaining the
fronthaul, CI, and power SOC constraints on the last column. Defining
\(\bm d_{k,Y}=\col(\bm D_k\bm p_Y,\sqrt{\Gamma_k}\sigma_v)\), the
SOC-strengthened SDR is
\begin{subequations}
\label{prob:sdr}
\begin{align}
  \calI_{\rm SDR}\triangleq
  \max_{\bm Y}\quad&
  \langle\widehat{\bm Q},\bm Y\rangle\\
  \operatorname{s.t.}\quad&
  \eqref{eq:rank-lift1}-\eqref{eq:rank-lift4}\text{ except }\rank(\bm Y)=1,\notag\\
  &\|\bm A_n\bm y_Y\|_2
    \le\sqrt{\alpha_{\rm dl}}\,
       (\bm e_n^p)^\TT\bm y_Y,\quad\forall n,
       \label{eq:sdr-soc-fh}\\
  &\|\bm d_{k,Y}\|_2
    \le(\bm a_{k,t}^{\nu})^\TT\bm y_Y,
    \quad\forall k,t,\nu,\\
  &\|\bm y_Y\|_2\le\sqrt{P_{\max}}.
  \label{eq:sdr-soc-power}
\end{align}
\end{subequations}
The added SOCs are redundant at rank one but effective after
relaxation: they preserve every original feasible point while excluding
high-rank moment solutions whose last columns violate the original
SOCs. Hence, the resulting bound is no looser than the standard
SDR bound and remains a valid global upper bound. Combining it with the
SLBM lower bound gives
\(\calI_{\rm SLBM}\le\calI^\star\le\calI_{\rm SDR}\), so their gap
provides an a posteriori bound on the suboptimality of the feasible
design.

%% file: sections/06_simulations.tex
\section{Simulation Results and Comparisons}
\label{sec:simulations}

In this section, we conduct simulations to validate the proposed
fronthaul-limited SLP design. Unless otherwise stated, we consider
\(L=2\) TXs with \(N_t=8\) antennas each, one \(M=32\)-antenna sensing
RX, and \(K=12\) single-antenna users. Communication channels follow
independent Rayleigh fading with distance-dependent path loss, while the
target channel follows \eqref{eq:G}. We use quadrature PSK (QPSK), \(T=K=12\),
\(\Gamma_k=10\) dB, and \(P_{\max}=45\) dBm. The downlink and uplink
fronthaul capacities are set equal, i.e., \(C^{\rm dl}=C^{\rm ul}=C\),
with default value \(C=80\) Mbps. All results are averaged over 200 random trials.

We compare four schemes: the proposed SLP-SLBM feasible design, the
SLP-SDR upper bound, and two BLP baselines based on
\cite{zhang2025networked}. Specifically, BLP-SDR follows the SDR-based
BLP model with fronthaul compression, while zero-forcing BLP (BLP-ZF) uses zero-forcing
communication beams with an additional sensing beam. All schemes are
evaluated under the same fronthaul and power constraints.

Fig.~\ref{fig:pmax-sweep} compares the sensing SINR versus the transmit-power budget
\(P_{\max}\) from 25 dBm to 50 dBm. The proposed SLP-SLBM design closely
tracks the SLP-SDR upper bound and outperforms BLP-SDR and BLP-ZF,
especially in the low-to-moderate power regime. This is because the
proposed SLP design steers known multiuser interference into CI regions
instead of suppressing it, allowing more available power to be used for
target illumination. As \(P_{\max}\) increases, this advantage becomes
less pronounced because the power constraint is relaxed for all schemes.

\begin{figure}[t]
  \centering
  \includegraphics[width=0.99\columnwidth]{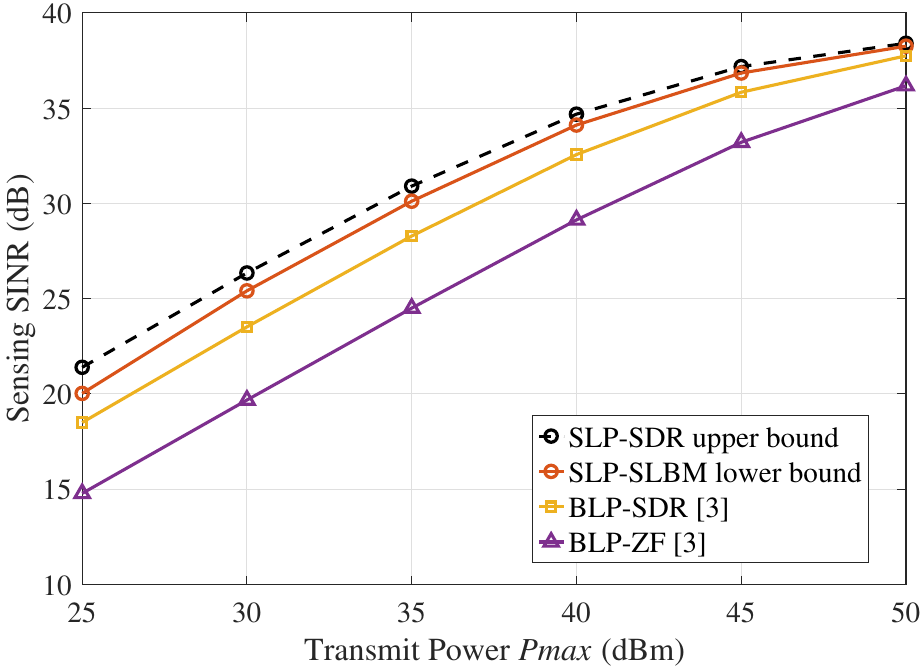}
  \vspace{-2.0ex}
  \caption{Sensing SINR versus transmit power.}
  \label{fig:pmax-sweep}
  \vspace{-2.0ex}
\end{figure}

Fig.~\ref{fig:resource-sweeps} compares the schemes under different fronthaul and antenna
resources. In Fig.~\ref{fig:resource-sweeps} (left), the advantage of SLP-SLBM over the BLP benchmarks grows with fronthaul capacity, since reduced compression
noise allows the SLP waveform to better exploit CI for sensing-oriented
shaping. In Fig.~\ref{fig:resource-sweeps} (right), SLP-SLBM remains close to the SLP-SDR upper bound and above the BLP benchmarks, showing that the proposed SLP design uses additional spatial degrees of freedom more effectively for CI control and target illumination.

\begin{figure}[t]
  \centering
  \includegraphics[width=0.99\columnwidth]{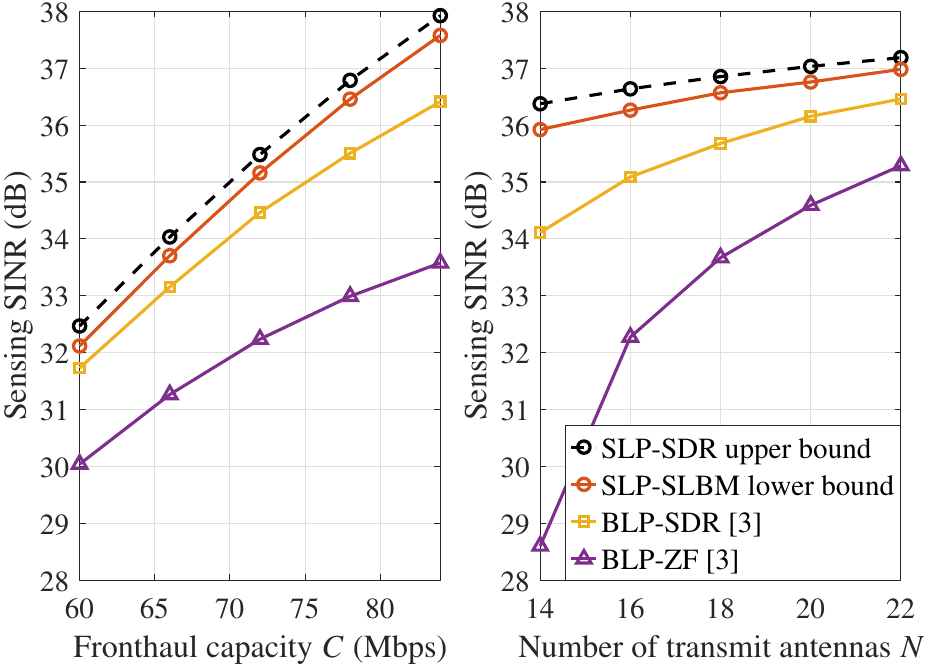}
  \vspace{-2.0ex}
  \caption{Sensing SINR versus fronthaul capacity and the total number
of transmit antennas.}
  \label{fig:resource-sweeps}
  \vspace{-2.0ex}
\end{figure}